\documentclass[conference]{IEEEtran}
\ifCLASSINFOpdf
\else
\fi
\makeatletter
\def\ps@IEEEtitlepagestyle{
  \def\@oddfoot{\mycopyrightnotice}
  \def\@evenfoot{}
}
\def\mycopyrightnotice{
  {\footnotesize
  \begin{minipage}{\textwidth}
  \centering
 Copyright~\copyright~2017 IEEE. Personal use of this material is permitted. Permission from IEEE must be obtained for all other uses, in any current or future media, including reprinting/republishing this material for advertising or promotional purposes, creating new collective works, for resale or redistribution to servers or lists, or reuse of any copyrighted component of this work in other works.
  \end{minipage}
  }
}

\usepackage[english]{babel}
\usepackage[T1]{fontenc}

\usepackage[utf8x]{inputenc}
\usepackage{cite}
\usepackage{graphicx}
\usepackage{color}
\usepackage{amssymb}
\usepackage{amsthm}
\usepackage{amsxtra}
\usepackage{amsmath}
\usepackage{multirow}
\usepackage{epsfig}
\usepackage{comment}
\usepackage{subfigure}
\usepackage{psfrag}
\usepackage{url}
\usepackage{textcomp}
\usepackage{enumerate}
\usepackage[printonlyused]{acronym}
\usepackage{epstopdf}
\usepackage{balance}
\usepackage{algorithm}
\usepackage{algpseudocode}
\usepackage{nomencl}
\usepackage{units}
\usepackage{pstricks, pst-plot, pst-grad, pstricks-add, pst-node, pstricks-add}
\algnewcommand{\Inputs}[1]{%
  \State \textbf{Inputs:}
   \hspace*{\algorithmicindent}\parbox[t]{.8\linewidth}{\raggedright #1}
}
\algnewcommand{\Initialize}[1]{%
  \State \textbf{Initialization:}
  \Statex \hspace*{\algorithmicindent}\parbox[t]{.8\linewidth}{\raggedright #1}
}
\algnewcommand{\Output}[1]{%
  \State \textbf{Output:}
   \hspace*{\algorithmicindent}\parbox[t]{.8\linewidth}{\raggedright #1}
}
\usepackage{pifont}
\usepackage{glossaries}

\newtheorem{theorem}{Theorem}

\newcommand{\Matrix}[1]{\mathbf{#1}}

\newrgbcolor{darkgreen}{0.2 0.6 0.2}
\newrgbcolor{darkred}{0.6 0.2 0.2}
\newrgbcolor{darkblue}{0.2 0.2 0.6}

\begin{document}
%\title{Joint Optimization of Communication and Computational resources}
\title{Device-centric Energy Optimization for Edge Cloud Offloading}


% \author{\IEEEauthorblockN{Shreya Tayade}
% % \IEEEauthorblockA{Institute of Wireless Communication and Navigation\\
% % University of Kaiserslautern, Germany\\
% % Email: tayade@eit.uni-kl.de}
%  \and
%  \IEEEauthorblockN{Peter Rost}
% % \IEEEauthorblockA{Twentieth Century Fox\\
% % Springfield, USA\\
% % Email: homer@thesimpsons.com}
% }


% \author{\IEEEauthorblockN{Shreya Tayade \IEEEauthorrefmark{1},
% Peter Rost\IEEEauthorrefmark{2},
% Andreas Maeder\IEEEauthorrefmark{3}, 
% Hans Schotten \IEEEauthorrefmark{3}}
% %\IEEEauthorblockA{\IEEEauthorrefmark{1}School of Electrical and Computer Engineering\\
%Telephone: (800) 555--1212, Fax: (888) 555--1212}

\author{
\IEEEauthorblockN{Shreya Tayade\IEEEauthorrefmark{1},
Peter Rost\IEEEauthorrefmark{2},
Andreas Maeder\IEEEauthorrefmark{2} and 
Hans D. Schotten\IEEEauthorrefmark{1}}
\IEEEauthorblockA{\IEEEauthorrefmark{1}University of Kaiserslautern,
Institute for Wireless Communications and Navigation,
Kaiserslautern, Germany\\
Email: \{tayade, schotten\}@eit.uni-kl.de}
\IEEEauthorblockA{\IEEEauthorrefmark{2}Nokia Bell Labs,
Munich, Germany\\
Email: \{peter.m.rost, andreas.maeder\}@nokia-bell-labs.com}
}

\maketitle

\begin{abstract}
A wireless system is considered, where, computationally complex algorithms are offloaded from user devices to an edge cloud server, for the purpose of efficient battery usage. The main focus of this paper is to characterize and analyze, the trade-off between the energy consumed for processing the data locally, and for offloading. An analytical framework is presented, that minimizes the in-device energy consumption, by providing an optimal offloading decision for multiple user devices. A closed form solution is obtained for the offloading decision. The solution also provides the amount of computational data that should be offloaded, for the given computational and communication resources. Consequently, reduction in the energy consumption is observed.
\end{abstract}
%A closed form solution is developed that provides an optimal offloading decision, along with the amount of computational data that should be offloaded. Consequently, reduction in the energy consumption is observed.
%\keywords{Cloud Offloading, computational complexity, communication complexity, energy optimization}

% \fi
\IEEEpeerreviewmaketitle
\section{Introduction}
% This demo file is intended to serve as a ``starter file''
% for IEEE conference papers produced under \LaTeX\ using
% IEEEtran.cls version 1.8b and later.
% \textcolor{red}{Why is this new and important problem\\
% What has been done before \\ 
% How does my research brings significant new understanding to the problem}
In recent years, many new services and use cases with a focus on Internet of Things (IoT),  such as smart city, factory automation and so on, have emerged for wireless communications. Most of these use cases are realized by deploying computationally complex algorithms on user devices with limited computational resources and battery capacity. For example, a surveillance drone executes complex image processing algorithms for object detection and tracking. Executing them may readily discharge the battery of these devices due to high energy consumption.\\
%In many of these use cases, complex algorithms are deployed on user devices with limited computational resources and battery capacity.    
An alternative solution is to offload these algorithms to a centralized server, which can be located in an edge cloud. This may reduce the device energy consumption, while simultaneously increasing the flexibility of deploying even more complex algorithms. Moreover, centralized processing is crucial in some cases such as factory automation, where robots need to collaborate, communicate, coordinate and synchronize for a given task.
%Furthermore, may limit the algorithm complexity due to limited computational resources and
% In the current industry scenario, most of the devices have in device implementation of algorithms. For some use cases it is beneficial, however in most of the cases the devices needs to communicate, collaborate, upgrade the algorithms, and serve the services for longer duration of time. This generates the need of edge cloud in the industry vertical.
%Also, some of the devices performing same function, but belonging to diff,erent manufacturer might have different algorithms implementation. In order to communicate among each other, a protocol mapping among these devices is required. The devices can communicate easily and seamlessly if the replica of these algorithms are placed in the edge cloud.
%Furthermore, some manufacturers are concerned regarding the security of their implementation (developed algorithms), that could be easily availed by executing the algorithms in the edge cloud. \\
%\subsection{Motivation}
%\begin{itemize}
%\item What is the problem?
The main challenge is to make the correct offloading decision, i.e., to assess the right criterion and threshold to offload an algorithm to the edge cloud. Even though the computational load on the device can be reduced by offloading, an additional communication load is introduced for transmitting the data to the edge cloud. Therefore, there exists a trade-off between communication load and computational load that user devices experience. To increase the energy efficiency of the user devices, it is necessary to take the offloading decision by analyzing this trade-off. The relevant parameters for this trade-off include communication and computational resources, algorithm's complexity, load condition on the cloud, device energy consumption, and delay constraints. 

\subsection{Related Work} 
Computation offloading is extensively studied recently \cite{Kumar, Kumar2013, mao2017mobile}. \cite{Kumar} provides a general overview addressing the circumstances under which offloading can save energy. The author has drawn some interesting conclusions, by analyzing the computational load and the available communication resources for a single user case. However, in practice, multiple users share the available resources, and hence, the analysis for multi-user scenario is necessary. Also, many energy minimizing techniques have been proposed in the literature for efficient computation offloading \cite{Zhang2013, XudongXiang2014, Cui2013,Zhang2016,You2017}. In \cite{Zhang2013}, an energy consumption is reduced by optimally scheduling data transmission over a wireless channel, and dynamically configuring the clock frequency of the local processor. Similarly, \cite{XudongXiang2014} presents an algorithm, based on stochastic dynamic programming, with an objective to energy efficiently schedule data transmission and link selection. A computational offloading problem was designed in \cite{Xchen}, based on the game theory approach, for multiple users considering a multi-channel interference environment. \cite{Zhang2016} provides an optimal computation offloading mechanisms in 5G heterogeneous environment. The approach is to effectively classify and prioritize the users, followed by optimally allocating the radio resources. \cite{You2017} also minimizes the energy consumption by optimal resource allocation for TDMA and OFDMA systems. The contributions in \cite{Sardellitti2015} and \cite{Sardellitti2014} deal with joint optimization of communication and computational resources for multiple users, so that the delay constraints are met. In contrast to optimally allocating resources, as in \cite{Zhang2016,You2017,Sardellitti2015,Sardellitti2014}, we evaluate the optimal offloading strategy for the allocated communication and computational resources. 

Apart from optimal resource allocation, for computational offloading, approaches like task partitioning and scheduling have been proposed in \cite{Kao,Wu2016}. In \cite{Kao}, the author presents an algorithm to partition a single task and optimally offload these partitioned task by analyzing their dependencies. A low complexity algorithm, that minimizes the device energy consumption by dynamically offloading a partitioned task, is designed with Lyapunov optimization in \cite{DongHuang2012}. The algorithms in \cite{Kao} and \cite{DongHuang2012} consider computational complexity to offload each partitioned task, but do not consider the effects of channel and availability of communication resources. The papers\cite{Kao}, \cite{DongHuang2012} and \cite{Kumar} lack the crucial analysis of the energy consumption for multi-user scenario, where the communication, and the edge cloud resources are shared by multiple users. 
\subsection{Contribution and outline of the paper}
This paper analyzes the trade-off between the energy consumption due to local processing, and offloading, in order to evaluate an optimal offloading decision. The optimal offloading decision is evaluated considering the effects of communication channel, load introduced at the edge cloud server by multiple users, computational complexity of the data processing algorithm, and availability of communication resources. We introduce a simple algorithm that not only provides the optimal offloading decision for multiple users, but also provides the optimum amount of computation data that should be offloaded. 
%%%%%%%%%%%%%%%%%%%%%%%%%%%%%%%%%%%%%%%%%%%%%%%%%%%%%%%%%%%%%%%%%%%%%%%%%%%%%%%%%%%%%%%%%%%%
In Section~\ref{sec: system:model}, we describe the system model, including an energy consumption model for the user devices considering the algorithmic computational complexity and the communication complexity for offloading. The energy optimization problem and the closed form solution is presented in Section~\ref{sec: sum.energy}. Finally, the results and conclusion are discussed in Section~\ref{sec:results} and \ref{sec: conclusion} respectively.

\section{System Model} \label{sec: system:model}
Consider $N$ user devices uniformly distributed in a circular area of radius $R$. In the center of the area, the base station is placed and co-located with an edge-cloud server. The base station
has knowledge of the channel condition of each user $i\in[1; N]$.
The edge cloud has a processor with a maximum computational capacity of $C_s$, and each user device has a maximum computational capacity of $C_{u}$, where $C_{s} \gg C_{u}$. 

\subsection{Data model}
In each time period $T$, every user device needs to process $D_i$ data bits, which may either be processed by the device itself or offloaded to the edge cloud. The share of data per user device, that is offloaded in time period $T$, is given by $0\leq \alpha\leq 1$. As shown in  Fig.~\ref{fig:system_model}, the data $D_i$ is composed of $L$ data blocks, each composed of $M$ data elements with $S$ bits, i.\,e., $D_i = L\cdot M\cdot S$. This corresponds, for instance, to an industrial automation scenario where a field-bus gateway receives $M$ data elements from $L$ connected sensors during each time period in order to perform an update of the automation schedule.
The data processing algorithm has the complexity class, given by the function $f_i(M)$, that defines the amount of computational complexity introduced on the user device with respect to the increase in the number of data elements. 

\subsection{Device computational complexity and energy consumption} 
The computational complexity generated at a user device, if all the data is processed locally is given by 
\begin{align}
C_\text{u,i} &= L \cdot \eta_i f_{i}(M),
	\label{eq:system.model:device:10}
\end{align}
with the proportionality constant $\eta_i$ that depends on the processor specifications, and represents the amount of computation cycles required to execute the algorithm, when the number of data elements $M$ is $1$. Consequently, the energy consumed by the user device depends on the number of computation cycles required to process $M$ data elements. The number of computational cycles further depends on the number and the type (read/write, memory access) of the operations involved in the algorithm. As the detailed analysis of operation-specific energy consumption for a particular algorithm is out of the scope of this paper, we represent the total energy consumption in terms of the computation complexity, as given in \cite{HopfnerTowardsStrategies}. If the average amount of energy consumed by the user device for a single computation cycle is $\epsilon_i$, then the total energy consumed $E_\text{u,i}$ on the user device during time period $T$ is given by
\begin{align}
E_\text{u,i} & = \epsilon_{i} \cdot C_\text{u,i}  =
\epsilon_{i} \cdot L \cdot \eta_i f_{i}(M). \label{eq: optimization:1}
\end{align}
%Computing these computationally complex algorithms, may drain large amount of the power from the user device. Therefore, in order to retain their battery life for a longer duration of time, these user devices offload the computation to the edge cloud over the wireless communication channel. To compute these algorithms on the edge cloud, it is necessary to transmit the data elements that are required for processing of the given algorithm. 

% only depends on the computational complexity but also on the 
% For an algorithm of the given complexity, the energy consumption not only depends on the number of CPU cycles required to execute the given algorithm but also on the memory and data access operations.
%The computational complexity of an algorithm is defined as the rate at which amount of computation increases with an increase in the number of data elements.  
% defined by function $f_i(M)$, if $M$ is the 
% These user devices offload the computationally complex algorithms to the edge cloud in order to retain their battery power for a longer duration of time. 
% The user devices
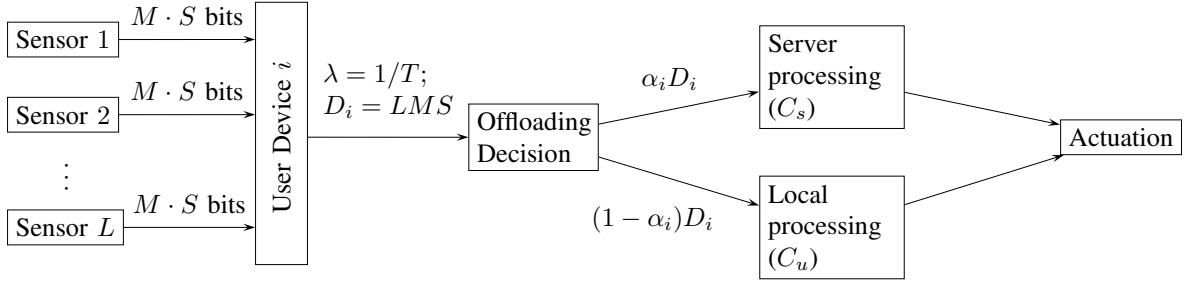
\begin{figure*}
  \centering
  \begingroup
\unitlength=1mm
\begin{picture}(155, 39)(0, 0)
\psset{xunit=1mm, yunit=1mm, linewidth=0.1mm}
\psset{arrowsize=2pt 3, arrowlength=1.6, arrowinset=.4}
%\psframe(0, 0)(155, 40)%
\rput(0, -3){%
	\rput[l](0, 35){\rnode{S1}{\psframebox{Sensor $1$}}}%	
	\rput[l](0, 25){\rnode{S2}{\psframebox{Sensor $2$}}}%	
    \rput[l](8, 17){\rput{90}{$\cdots$}}%
	\rput[l](0, 10){\rnode{SL}{\psframebox{Sensor $L$}}}%	
    \rput[l](40, 5){\pnode(0, 17){MUXout}%
    	\pnode(-7, 30){MuxIn1}%
    	\pnode(-7, 20){MuxIn2}%
    	\pnode(-7, 5){MuxInL}%
    	\rnode{MUX}{\rput{90}{\psframe(35, 7)\rput[c](17, 3.5){User Device $i$}}}}%
    \rput[c](70, 22){\rnode{Offload}{\psframebox{\parbox[c]{1.5cm}{Offloading Decision}}}}%
    \rput[l](100, 10){\rnode{LocalProcess}{\psframebox{\parbox[c]{1.7cm}{Local\\processing ($C_u$)}}}}%
    \rput[l](100, 30){\rnode{ServerProcess}{\psframebox{\parbox[c]{1.7cm}{Server\\processing ($C_s$)}}}}%
    \rput[l](140, 22){\rnode{Actuation}{\psframebox{Actuation}}}%
}
    \ncline{->}{MUXout}{Offload}\naput{$\begin{array}{l}\lambda = 1/T;\\D_i = LMS\end{array}$}
    \ncline{->}{Offload}{ServerProcess}\naput{$\alpha_iD_i$}
    \ncline{->}{Offload}{LocalProcess}\nbput{$(1-\alpha_i)D_i$}
    \ncline{->}{ServerProcess}{Actuation}
    \ncline{->}{LocalProcess}{Actuation}
    \ncline{->}{S1}{MuxIn1}\naput{$M\cdot S$ bits}
    \ncline{->}{S2}{MuxIn2}\naput{$M\cdot S$ bits}
    \ncline{->}{SL}{MuxInL}\naput{$M\cdot S$ bits}
\end{picture}
\endgroup
  \caption{Data processing model}
  \label{fig:system_model}
\end{figure*}
%\begin{figure}[!t]
%\centering
%\includegraphics[width=\columnwidth]{figures/SystemModel.png}
%\caption{System Model}
%\label{fig:system_model}
%\textcolor{red}{I think that the right side is a bit misleading because the actual offloading is not time-variant but stays constant over individual time slots. In addition, I would show the picture for only one device but then including the three phases of uplink, processing, and downlink.}
%\end{figure}
% \begin{figure}[!t]
% \centering
% \includegraphics[width=2.5in]{figuresv2.1/system_model.png}
% \caption{System Model}
% \label{fig_sim}
% \end{figure}
%The communication between $i^\text{th}$ user device and the edge cloud takes place over the wireless communication channel.
% The pathloss exponent $\beta$ represents the amount of pathloss experience by the user device in the uplink transmission. Each user device $i$ can transmit with the transmit power $P_\text{tr,i}$ in uplink. 
\subsection{Channel model}
The user devices transmit the data to the edge cloud using frequency division multiple access (FDMA), i.\,e., the carrier bandwidth $B$ is distributed equally among all user devices, such that each user device uses bandwidth $B_i$ distributed across $N_\text{RB, i}$ resource blocks (RBs). 
The effects of opportunistic scheduler are not considered in this paper for the sake of brevity. Each RB corresponds to a bandwidth of 180 kHz and time-slot duration $T_\text{slot}$ = 0.5 ms \cite{3GPPTS36.2112010}, i.\,e., $B_i = N_\text{RB,i} \cdot \unit[180]{kHz}$. 

The received signal-to-noise ratio (SNR) for user device $i$ is given by
\begin{equation}
\gamma_\text{ul,i} = \frac{P_\text{r,i}}{N_0 B_i},
\end{equation}
with the received signal power $P_\text{r, i}$ and the noise power spectral density $N_0$.
The $i^\text{th}$ user device is located at a distance $d_i$ from the cell center, hence, the received power is given by
\begin{equation}
  P_\text{r, i} = P_\text{tr, i} \cdot G \left[ \frac{d_0}{d_i}\right]^\beta,	
   \label{eq:system:model:2} 
\end{equation}
with the pathloss exponent $\beta$, transmit power $P_\text{tr, i}$, reference distance $d_0$, and $G = \left(\frac{\lambda}{4\pi d_0}\right)^2$ being an attenuation constant for free-space path-loss. We assume that $G$ is known at the base station.\\
Given the received SNR, the spectral efficiency is given by $r_i = \log_2(1 + \gamma_\text{ul, i}) \leq 6$ bps/Hz, which is the maximum spectral efficiency achievable in 3GPP LTE \cite{3GPP2009a}.\\%\cite{201039}.\\ 
%  The maximum data rate $r_\text{max}^i$ achieved by the user device per time slot is
% \begin{equation}
% r_\text{max}^i = \frac{N_\text{RB} N_\text{RE} \log_2{(M_{CS})}}{T_\text{slot}}
% \end{equation}
\subsection{Transmission energy model}
The user device has to offload $D_i$ bits to the edge cloud in the time interval $T$, i.\,e., the spectral efficiency in the time interval has to satisfy the equation 
\begin{equation}
D_i = T \cdot B_i \cdot \log_2\left(1 + \frac{P_\text{r,i}}{N_0 B_i}\right).
\end{equation}
Hence, the required receive signal power in order to transfer all $D_i$ bits to the edge cloud in the given time period $T$ is 
\begin{equation} 
P_\text{r, i}\stackrel{!}{=}  \left(2^{D_i/(B_i T)} - 1\right) N_0 B_i \label{eq:system:model:1}
\end{equation}
Using \eqref{eq:system:model:2}, the required transmit power is given by
\begin{equation} \label{eq:system:model:3}
P_\text{tr, i} \stackrel{!}{=} \frac{\left(2^{D_i/(B_i \cdot T)} - 1\right)}{G} \cdot \left [\frac{d_i}{d_0}\right]^{\beta} \cdot N_0 B_i,
\end{equation}
which is upper limited by $P_\text{tr, i}\leq P_\text{tr, max}$ \cite{3GPP2009a}.
Hence, the energy consumed by the $i^\text{th}$ user device to transmit its $D_i$ data bits is given by 
\begin{align} 
     E_\text{tr,i} & = P_\text{tr,i} \cdot T \\
 & = \frac{\left(2^{D_i/(B_i \cdot T)} - 1\right)}{G} \cdot \left [\frac{d_i}{d_0}\right]^{\beta} \cdot N_0 B_i \cdot T .\label{eq:system:model:4}
\end{align}
The energy consumed for transmitting the data to the edge cloud is largely impacted by the pathloss, allocated bandwidth, and the amount of data that is require to be offloaded.% In our system, we assumed an unlimited power control for the uplink which is a non-practical assumption. However, applying limited power control will not change the result trends. 

\subsection{Energy consumption at the user device}
The previous model is now extended by taking into account the possibility of offloading only a share $\alpha_i D_i$, $0\leq\alpha_i\leq 1$, of the overall data. Accordingly, the models in 
~\eqref{eq: optimization:1} and ~\eqref{eq:system:model:4} are modified to be
\begin{equation}\label{eq:system:model:energy:2}
E_\text{u,i}(\alpha_i) =  (1- \alpha_i) L \times\epsilon_{i} \times \eta_i f_{i}(M)
\end{equation}
and  
\begin{equation} \label{eq:system:model:energy:1}
E_\text{tr,i}(\alpha_i)  =  \frac{\left(2^{\alpha_i D_i/(B_i T)} - 1\right)}{G} \cdot \left [\frac{d_i}{d_0}\right]^{\beta} \cdot N_0 B_i \cdot T
\end{equation}
respectively. The total energy consumption of the user device $i$ can be given as 
\begin{align}
E_\text{sum,i}(\alpha_i) = E_\text{tr,i}(\alpha_i) + E_\text{u,i}(\alpha_i) . \label{eq:system:model:energy:10}
\end{align}
The static energy consumption of the user device during idle time is fixed, and hence can be neglected in the model for making an offloading decision. 
% also mention receiving power to decode the results depends if the results are sent back in case of factory actuators and all but in use case like drones the results are not required to send back.
\subsection{Edge cloud processing} 
Similar to the computational complexity introduced on the user device by the algorithm,  the computational complexity $C_\text{serv, i}$ is also introduced on the edge cloud, if the computation is offloaded.  However, the proportionality constant $\eta_{s}$ for the edge cloud is different, and depends upon its processor characteristics. 
The computational complexity on the edge-cloud is
\begin{equation}
C_\text{serv, i} = \eta_{s} \cdot f_{i}(M).
\end{equation}
Given the edge cloud processor's capacity $C_s$, the maximum number of computation cycles that the server can schedule in time period $T_\text{pr}$ is defined by $C_\text{s,max} = C_s \cdot T_\text{pr}$. We assume that $T_\text{pr}\ll T$ because one edge cloud server would need to process the data of the user devices from more than one cell.

\section{Sum Energy Optimization}   \label{sec: sum.energy}
\subsection{Problem formulation:}
As discussed in Section~\ref{sec: system:model}, we consider the energy consumed for in-device data processing, as well as for offloading the data to the edge cloud. The optimization problem is device-centric, and designed to minimize the total energy consumption $E_\text{sum,i}$, for all  $N$ user devices by offloading an optimal share of data, as given by the set of decision variables $\mathcal{A} = \{\alpha_1, \dots \alpha_N\}$. If $\alpha_i$ is $0$, no data is offloaded to the cloud, whereas if $\alpha_i$ is $1$, all the data is offloaded to the edge cloud.
The optimization problem is given as:
 \begin{align} \label{eq:opt.sum.energy.1}
 \mathcal{A}' = \text{arg}\min_{\forall \mathcal{A}\in\mathbb{R}^N} &\: \sum_{i=1}^N \: E_\text{sum,i}(\alpha_i) \nonumber  \\
 \text{s.t}  \quad  & \sum\limits_i^N L \cdot \alpha_i \cdot C_\text{serv, i} \leq C_\text{s,max} \nonumber \\
& 0 \leq \alpha_i \leq 1 ,
 \end{align}
%  \text{where}
%  \begin{align}
%  \mathcal{A} = \{\alpha_1, \dots \alpha_N\}
%  \end{align}
%   \begin{align}
%  	\mathcal{A}_k & = \{ \alpha_1, ... \alpha_N\}; 0 \leq \alpha_i \leq 1 \\
%     \mathcal{A} & = \bigcup\limits_{k\in [1; 2^N]} \mathcal{A}_k
%   \end{align}
 % is minimized.such that schedules N user devices on the cloud server by obtaining the optimal assignment variable set $\mathcal{A}_k  = \{ \alpha_1, ... \alpha_N\}$,  
The limiting constraint for offloading is that the total amount of required computational cycles to process the offloaded computation, should not exceed the maximum computational cycles $C_\text{s,max}$, that the server can provide in the given time period $T_\text{pr}$.
We further distinguish \emph{state-full} (SF) and \emph{state-less} (SL) offloading. In the case of SF offloading, every user device either offloads all the computation to the edge cloud or does not offload at all for a given period $T$. The value of  offloading parameter is $\alpha_i = \left\lbrace 0,1 \right\rbrace$. This corresponds to the case where the processing algorithm cannot be divided due to mutual data dependencies.
In the case of SL offloading, the user device is allowed to offload any partition of the data processing, i.\,e., $\alpha_i$ is therefore relaxed in the optimization problem and it lies between $[0;1]$. This corresponds to the case mentioned earlier, where $L$ sensors provide data to a gateway device, which processes these data independently.
\subsection{Solution to optimization problem}
This optimization problem is solved using Lagrange's Duality Theorem and by applying Karush-Kuhn-Tucker (KKT) conditions. The objective function is given as
\begin{eqnarray}
    	\mathcal{L}(\alpha_i, \nu,\psi) & = & \sum\limits_i^N\left(E_\text{u,i}(\alpha_i) + E_\text{tr,i}(\alpha_i)\right) \nonumber \\ 
        	 & & {+}\: \nu \left(\sum\limits_i^N L \cdot \alpha_i \cdot C_\text{serv,i} \leq C_\text{s,max}\right) \nonumber \\ 
            & & {-}\: \text{tr}\left[\Matrix{\Psi}\text{diag}(\alpha_i)\right]   
            \label{eq:opt.sum.energy.10}
\end{eqnarray}
where $\nu$ and $\psi$ are the Lagrange multipliers. The solution to this optimization problem is very similar to the water-filling algorithm and drives us towards two theorems stated below.

\begin{theorem}\label{theorem:optimization.problem:10}
The optimum offloading parameter $\alpha_i$ for the $i^\text{th}$ user device is given by
\begin{equation}
  \alpha_i = \left(\frac{1}{r_i}\log_2\left(\frac{1}{K_i}\left[E_\text{u,i} - \nu \: C_\text{serv,i} \right]\right)\right)^+
\end{equation}
with $r_i$= $D_i/(B_iT)$, a constant    $K_i=\left( \text{ln}(2) \left [\frac{d_i}{d_0}\right]^{\beta} N_0 D_i / G \right)$, and the Lagrangian parameter '$\nu$' defines the offloading threshold for the user device.
\end{theorem}
\begin{proof}
See Appendix \ref{subsec:proof1}.
\end{proof}
The Lagrangian parameter $\nu$ is derived through an iterative method. For an overloaded system, where the cloud server capacity is not able to serve the computational load coming from all the users, i.e. $\sum\limits_i^N L \cdot C_\text{serv,i} \geq C_\text{s,max}$, the threshold is increased stepwise, until the condition in \eqref{eq:opt.energy.sum.111} is satisfied. With this action, the user devices that save less energy by offloading, out of all the user devices, are not allowed to offload anymore. The corresponding constraints on $\nu$ are defined in the following theorem.
\begin{theorem}
Given the solution to the optimization problem in Theorem \ref{theorem:optimization.problem:10}, the threshold '$\nu$' is bounded by
\begin{eqnarray}
 \max\limits_{i: \alpha_i>0} \left( \left[\frac{E_\text{u,i} - K_i2^{r_i}}{C_\text{serv,i}}\right]\right)^+  
 \leq  \nu \leq 
 \min\limits_{i: \alpha_i>0}\left[\frac{E_\text{u,i} - K_i}{C_\text{serv,i}}\right]. 
\end{eqnarray}
\end{theorem}
Note that the upper and lower bounds on $\nu$  holds only for the user devices, with $\alpha_i \neq 0$.
\begin{proof}
See Appendix \ref{subsec:proof2}
\end{proof}

\section{Results and Discussion} \label{sec:results}
\subsection{Performance metrics}
\paragraph{Sum Energy}
The performance of SF and SL offloading is evaluated by comparing the total optimized energy i.\,e.  $E_\text{sum}(\mathcal{A}') = \sum_{\alpha_i\in\mathcal{A}} \: E_\text{sum,i}(\alpha_i)$, with the total energy consumed when no user device offloads the data processing, and when all the user devices completely offload the processing. The energy per user device $E_\text{sum,i}(\alpha_i)$ is given in (\ref{eq:system:model:energy:10}), where $\alpha_i$ is determined according to Theorem \ref{theorem:optimization.problem:10}.
The total energy consumption in the case that no user device offloads ($\forall i \in [1; \dots N]: \alpha_i = 0$) is given by 
\begin{equation}
E_\text{sum}(0) = \sum\limits_i^N \: E_\text{u,i}.
\end{equation}
Whereas, when every user device offloads all the data, the total energy consumption is given by
\begin{equation}
E_\text{sum}(1) = \sum\limits_i^N \:  E_\text{tr,i}(1).
\end{equation}

\paragraph{Offloading Percentage}
The offloading percentage is the ratio of total offloaded data processing for all user devices to the total data processing of the system, and is given by 
\begin{align}
\Lambda =  \frac{ \sum\limits_i^N \alpha_i \cdot D_i}{ \sum\limits_i^N D_i}.
\end{align}

\subsection{Performance depending on path-loss}
%\section{Simulation Parameters}
\begin{table}[t]
% \renewcommand{\arraystretch}{1.3}
% if using array.sty, it might be a good idea to tweak the value of
% \extrarowheight as needed to properly center the text within the cells
  \caption{Simulation Parameters}
  \label{tb:simulation:parameters}
  \centering
  \begin{tabular}{|c||c||c||c|}
    \hline 
    Variable & Value & Variable & Value \\ \hline
    \hline 
    $C_{s}$ & 200 MHz & $N$ & 50 Users \\
    \hline
    L & 10 & $M$ & 60 data elements \\
    \hline
    $BW$ & 10 MHz & S & 8 bits\\
    \hline
    $T_\text{pr}$ & 1 ms & $T$ & 20 ms\\
    \hline
        $d_0$ & 200 m & $R$ & 800 m \\
    \hline
    $\epsilon_i$ & $5e{-6}$ mJ & $\eta_i$ & 100 cycles \\
    \hline
     $f_i(M)$ & $M$  & $\eta_s$ & 1 cycle \\
    \hline
  \end{tabular}
\end{table}
% % 200 MHz processor   N = 20 Users
% RTT = 1 sec;
% L = 10
% input data size N = 100, with 32 bit encoding
% M = 32*100
% linear complexity
%  \eta_i = 10 % 10 cpu cycles for M = 1 
%  \epsilon_i = 5*10^-6 mJoules
% underloaded system
% N_o = 10^-9 mwatts % noise power
% BW = 10 MHz
% d_0 = 200 meters
% Cell radius = 1km
%%%%%% Pathloss variation %%%%%%%%%%%%%%%%%%%%%%%%%%%%%%%%%%%
% N = 20 Users
% Cs_max = 200*10^6 % 200 MHz processor 
%RTT = 1 sec;
% L = 10
% input data size N = 100, with 32 bit encoding
%M = 32*100
% linear complexity
%  \eta_i = 10 % 10 cpu cycles for M = 1 
%  \epsilon_i = 5*10^-6 mJoules
% underloaded system
% N_o = 10^-9 mwatts % noise power
% BW = 10 MHz
% d_0 = 200 meters
% Cell radius = 1km
%
Fig.~\ref{fig:results:offloading:1} shows the optimal offloading percentage $\Lambda$ for $N=50$ depending on different pathloss conditions. Two scenarios are assumed, with the availability of $\unit[100]{\%}$ and $\unit[10]{\%}$ of the cloud server capacity $C_s$. Where, $C_s = \unit[200]{MHz}$, to provide sufficient processing capacity for higher values of $M$ (discussed in the next subsection).

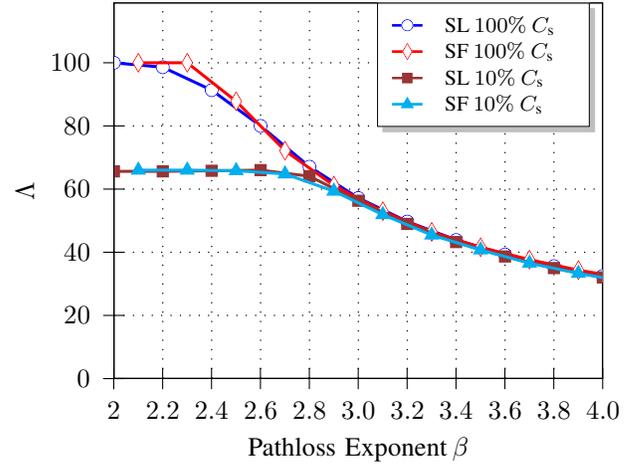
\begin{figure}[tb]
	\centering
    \begingroup
\unitlength=1mm
\psset{xunit=32.50000mm, yunit=0.42017mm, linewidth=0.1mm}
\psset{arrowsize=2pt 3, arrowlength=1.4, arrowinset=.4}\psset{axesstyle=frame}
\begin{pspicture}(1.53846, -38.08000)(4.00000, 119.00000)
\rput(-0.06154, -11.90000){%
\psaxes[subticks=0, labels=all, xsubticks=1, ysubticks=1, Ox=2, Oy=0, Dx=0.2, Dy=20]{-}(2.00000, 0.00000)(2.00000, 0.00000)(4.00000, 119.00000)%
\multips(2.20000, 0.00000)(0.20000, 0.0){9}{\psline[linecolor=black, linestyle=dotted, linewidth=0.2mm](0, 0)(0, 119.00000)}
\multips(2.00000, 20.00000)(0, 20.00000){5}{\psline[linecolor=black, linestyle=dotted, linewidth=0.2mm](0, 0)(2.00000, 0)}
\rput[b](3.00000, -26.18000){$\text{Pathloss Exponent} \: \beta $}
\rput[t]{90}(1.60000, 59.50000){$\Lambda$}
\psclip{\psframe(2.00000, 0.00000)(4.00000, 119.00000)}
\psline[linecolor=blue, plotstyle=curve, linewidth=0.4mm, showpoints=true, linestyle=solid, linecolor=blue, dotstyle=o, dotscale=1.2 1.2, linewidth=0.4mm](2.00000, 100.00000)(2.20000, 98.58797)(2.40000, 91.33355)(2.60000, 80.00576)(2.80000, 67.08496)(3.00000, 57.18920)(3.20000, 49.71945)(3.40000, 43.88725)(3.60000, 39.30643)(3.80000, 35.60574)(4.00000, 32.52836)
\psline[linecolor=red, plotstyle=curve, linewidth=0.4mm, showpoints=true, linestyle=solid, linecolor=red, dotstyle=diamond, dotscale=1.2 1.2, linewidth=0.4mm](2.10000, 100.00000)(2.30000, 100.00000)(2.50000, 87.80200)(2.70000, 72.16800)(2.90000, 61.11800)(3.10000, 52.99600)(3.30000, 46.46600)(3.50000, 41.57600)(3.70000, 37.67200)(3.90000, 34.33400)(4.10000, 31.71000)
\psline[linecolor=darkred, plotstyle=curve, linewidth=0.4mm, showpoints=true, linestyle=solid, linecolor=darkred, dotstyle=square*, dotscale=1.2 1.2, linewidth=0.4mm](2.00000, 65.61953)(2.20000, 65.62332)(2.40000, 65.79917)(2.60000, 66.00300)(2.80000, 64.07084)(3.00000, 56.30901)(3.20000, 48.97089)(3.40000, 43.23642)(3.60000, 38.66279)(3.80000, 34.99841)(4.00000, 32.03015)
\psline[linecolor=cyan, plotstyle=curve, linewidth=0.4mm, showpoints=true, linestyle=solid, linecolor=cyan, dotstyle=triangle*, dotscale=1.2 1.2, linewidth=0.4mm](2.10000, 66.00000)(2.30000, 66.00000)(2.50000, 65.82400)(2.70000, 64.79800)(2.90000, 59.34000)(3.10000, 51.90800)(3.30000, 45.44400)(3.50000, 40.73200)(3.70000, 36.54000)(3.90000, 33.32400)(4.10000, 30.60600)
\endpsclip
\psframe[linecolor=black, fillstyle=solid, fillcolor=white, shadowcolor=lightgray, shadowsize=1mm, shadow=true](3.07692, 79.73000)(3.94615, 119.00000)
\rput[l](3.35385, 111.86000){\footnotesize{$\text{SL} \: 100 \% \: C_\text{s}$}}
\psline[linecolor=blue, linestyle=solid, linewidth=0.3mm](3.13846, 111.86000)(3.26154, 111.86000)
\psline[linecolor=blue, linestyle=solid, linewidth=0.3mm](3.13846, 111.86000)(3.26154, 111.86000)
\psdots[linecolor=blue, linestyle=solid, linewidth=0.3mm, dotstyle=o, dotscale=1.2 1.2, linecolor=blue](3.20000, 111.86000)
\rput[l](3.35385, 103.53000){\footnotesize{$\text{SF} \: 100 \% \: C_\text{s}$}}
\psline[linecolor=red, linestyle=solid, linewidth=0.3mm](3.13846, 103.53000)(3.26154, 103.53000)
\psline[linecolor=red, linestyle=solid, linewidth=0.3mm](3.13846, 103.53000)(3.26154, 103.53000)
\psdots[linecolor=red, linestyle=solid, linewidth=0.3mm, dotstyle=diamond, dotscale=1.2 1.2, linecolor=red](3.20000, 103.53000)
\rput[l](3.35385, 95.20000){\footnotesize{$\text{SL}\: 10 \% \: C_\text{s}$}}
\psline[linecolor=darkred, linestyle=solid, linewidth=0.3mm](3.13846, 95.20000)(3.26154, 95.20000)
\psline[linecolor=darkred, linestyle=solid, linewidth=0.3mm](3.13846, 95.20000)(3.26154, 95.20000)
\psdots[linecolor=darkred, linestyle=solid, linewidth=0.3mm, dotstyle=square*, dotscale=1.2 1.2, linecolor=darkred](3.20000, 95.20000)
\rput[l](3.35385, 86.87000){\footnotesize{$\text{SF}\: 10 \% \: C_\text{s}$}}
\psline[linecolor=cyan, linestyle=solid, linewidth=0.3mm](3.13846, 86.87000)(3.26154, 86.87000)
\psline[linecolor=cyan, linestyle=solid, linewidth=0.3mm](3.13846, 86.87000)(3.26154, 86.87000)
\psdots[linecolor=cyan, linestyle=solid, linewidth=0.3mm, dotstyle=triangle*, dotscale=1.2 1.2, linecolor=cyan](3.20000, 86.87000)
}\end{pspicture}
\endgroup
%\endinput

    \caption{Offloading with pathloss variation}
    \label{fig:results:offloading:1}
\end{figure}
In Fig.~\ref{fig:results:offloading:1}, for the lower path-loss, $\beta\leq2.2$, $\unit[100]{\%}$ of the data processing is offloaded to the edge-cloud.
This illustrates that if all user devices have good channel conditions, they can minimize their energy-consumption by offloading to the edge cloud. As $\beta$ increases, the offloading percentage drops, as some user devices experience high channel attenuation. This results in an increase of the transmission energy required for offloading, as compared to the energy consumed for computation in the device itself.

In the second scenario with $\unit[10]{\%}C_s$, the edge cloud cannot simultaneously support offloading from all the users. Therefore, even though some users would prefer offloading, only a part of the data processing is carried out by the edge cloud. Hence, the maximum data processing supported by the edge cloud, does not exceed $\unit[65]{\%}$ of the total computation. The path-loss effects are prominent at $\beta>2.6$, and converges with $\unit[10]{\%}C_s$ scenario. This occurs due to a high number of user devices experiencing high channel attenuation, and hence, refrain from offloading. This illustrates that the edge cloud server capacity is not the limiting factor anymore.
%, as more number of users where the offloading percentage for both the scenarios are similar.  
%Hence, for $\unit[10]{\%}C_\text{s,max}$ about $\unit[65]{\%}$ of the maximum required server processing can be supported at the edge cloud.
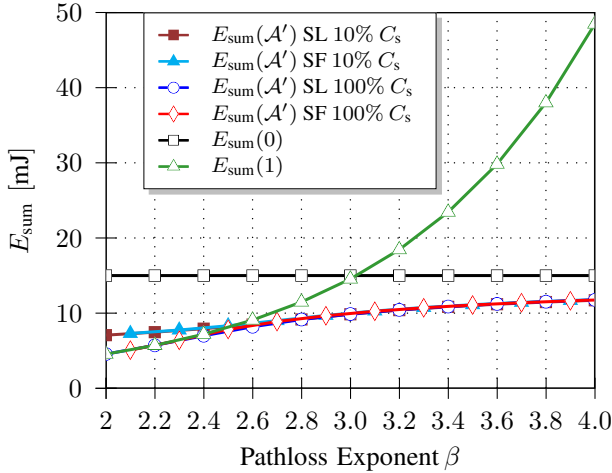
\begin{figure}[tb]
	\centering
    \begingroup
\unitlength=1mm
\psset{xunit=32.50000mm, yunit=1.00000mm, linewidth=0.1mm}
\psset{arrowsize=2pt 3, arrowlength=1.4, arrowinset=.4}\psset{axesstyle=frame}
\begin{pspicture}(1.53846, -16.00000)(4.00000, 50.00000)
\rput(-0.06154, -5.00000){%
\psaxes[subticks=0, labels=all, xsubticks=1, ysubticks=1, Ox=2, Oy=0, Dx=0.2, Dy=10]{-}(2.00000, 0.00000)(2.00000, 0.00000)(4.00000, 50.00000)%
\multips(2.20000, 0.00000)(0.20000, 0.0){9}{\psline[linecolor=black, linestyle=dotted, linewidth=0.2mm](0, 0)(0, 50.00000)}
\multips(2.00000, 10.00000)(0, 10.00000){4}{\psline[linecolor=black, linestyle=dotted, linewidth=0.2mm](0, 0)(2.00000, 0)}
\rput[b](3.00000, -11.00000){$\text{Pathloss Exponent} \: \beta $}
\rput[t]{90}(1.60000, 25.00000){$E_\text{sum} \: \left[ \text{mJ} \right]$}
\psclip{\psframe(2.00000, 0.00000)(4.00000, 50.00000)}
\psline[linecolor=darkred, plotstyle=curve, linewidth=0.4mm, showpoints=true, linestyle=solid, linecolor=darkred, dotstyle=square*, dotscale=1.2 1.2, linewidth=0.4mm](2.00000, 7.07256)(2.20000, 7.47052)(2.40000, 7.93380)(2.60000, 8.47231)(2.80000, 9.16365)(3.00000, 9.85116)(3.20000, 10.40255)(3.40000, 10.83736)(3.60000, 11.18787)(3.80000, 11.47529)(4.00000, 11.71394)
\psline[linecolor=cyan, plotstyle=curve, linewidth=0.4mm, showpoints=true, linestyle=solid, linecolor=cyan, dotstyle=triangle*, dotscale=1.2 1.2, linewidth=0.4mm](2.10000, 7.29736)(2.30000, 7.75566)(2.50000, 8.32108)(2.70000, 9.00335)(2.90000, 9.69977)(3.10000, 10.28788)(3.30000, 10.75305)(3.50000, 11.12612)(3.70000, 11.43335)(3.90000, 11.68500)(4.10000, 11.89923)
\psline[linecolor=blue, plotstyle=curve, linewidth=0.4mm, showpoints=true, linestyle=solid, linecolor=blue, dotstyle=o, dotscale=1.2 1.2, linewidth=0.4mm](2.00000, 4.55735)(2.20000, 5.71412)(2.40000, 6.98991)(2.60000, 8.18783)(2.80000, 9.15357)(3.00000, 9.87874)(3.20000, 10.43760)(3.40000, 10.87763)(3.60000, 11.23125)(3.80000, 11.52090)(4.00000, 11.76108)
\psline[linecolor=red, plotstyle=curve, linewidth=0.4mm, showpoints=true, linestyle=solid, linecolor=red, dotstyle=diamond, dotscale=1.2 1.2, linewidth=0.4mm](2.10000, 5.07830)(2.30000, 6.38792)(2.50000, 7.83054)(2.70000, 8.87976)(2.90000, 9.65380)(3.10000, 10.24800)(3.30000, 10.71289)(3.50000, 11.08217)(3.70000, 11.38623)(3.90000, 11.63911)(4.10000, 11.85083)
\psline[linecolor=black, plotstyle=curve, linewidth=0.4mm, showpoints=true, linestyle=solid, linecolor=black, dotstyle=square, dotscale=1.2 1.2, linewidth=0.4mm](2.00000, 15.00000)(2.20000, 15.00000)(2.40000, 15.00000)(2.60000, 15.00000)(2.80000, 15.00000)(3.00000, 15.00000)(3.20000, 15.00000)(3.40000, 15.00000)(3.60000, 15.00000)(3.80000, 15.00000)(4.00000, 15.00000)
\psline[linecolor=darkgreen, plotstyle=curve, linewidth=0.4mm, showpoints=true, linestyle=solid, linecolor=darkgreen, dotstyle=triangle, dotscale=1.2 1.2, linewidth=0.4mm](2.00000, 4.54698)(2.20000, 5.71345)(2.40000, 7.19553)(2.60000, 9.08090)(2.80000, 11.48205)(3.00000, 14.54343)(3.20000, 18.45066)(3.40000, 23.44238)(3.60000, 29.82561)(3.80000, 37.99555)(4.00000, 48.46119)
\endpsclip
\psframe[linecolor=black, fillstyle=solid, fillcolor=white, shadowcolor=lightgray, shadowsize=1mm, shadow=true](2.15385, 26.50000)(3.35231, 50.00000)
\rput[l](2.43077, 47.00000){\footnotesize{$ E_\text{sum}(\mathcal{A}')\: \text{SL}\: 10 \% \: C_\text{s}$}}
\psline[linecolor=darkred, linestyle=solid, linewidth=0.3mm](2.21538, 47.00000)(2.33846, 47.00000)
\psline[linecolor=darkred, linestyle=solid, linewidth=0.3mm](2.21538, 47.00000)(2.33846, 47.00000)
\psdots[linecolor=darkred, linestyle=solid, linewidth=0.3mm, dotstyle=square*, dotscale=1.2 1.2, linecolor=darkred](2.27692, 47.00000)
\rput[l](2.43077, 43.50000){\footnotesize{$ E_\text{sum}(\mathcal{A}')\:\text{SF} \: 10 \% \: C_\text{s}$}}
\psline[linecolor=cyan, linestyle=solid, linewidth=0.3mm](2.21538, 43.50000)(2.33846, 43.50000)
\psline[linecolor=cyan, linestyle=solid, linewidth=0.3mm](2.21538, 43.50000)(2.33846, 43.50000)
\psdots[linecolor=cyan, linestyle=solid, linewidth=0.3mm, dotstyle=triangle*, dotscale=1.2 1.2, linecolor=cyan](2.27692, 43.50000)
\rput[l](2.43077, 40.00000){\footnotesize{$ E_\text{sum}(\mathcal{A}')\:\text{SL} \: 100 \% \: C_\text{s}$}}
\psline[linecolor=blue, linestyle=solid, linewidth=0.3mm](2.21538, 40.00000)(2.33846, 40.00000)
\psline[linecolor=blue, linestyle=solid, linewidth=0.3mm](2.21538, 40.00000)(2.33846, 40.00000)
\psdots[linecolor=blue, linestyle=solid, linewidth=0.3mm, dotstyle=o, dotscale=1.2 1.2, linecolor=blue](2.27692, 40.00000)
\rput[l](2.43077, 36.50000){\footnotesize{$ E_\text{sum}(\mathcal{A}')\:\text{SF}\: 100 \% \: C_\text{s}$}}
\psline[linecolor=red, linestyle=solid, linewidth=0.3mm](2.21538, 36.50000)(2.33846, 36.50000)
\psline[linecolor=red, linestyle=solid, linewidth=0.3mm](2.21538, 36.50000)(2.33846, 36.50000)
\psdots[linecolor=red, linestyle=solid, linewidth=0.3mm, dotstyle=diamond, dotscale=1.2 1.2, linecolor=red](2.27692, 36.50000)
\rput[l](2.43077, 33.00000){\footnotesize{$ E_\text{sum}(0)$}}
\psline[linecolor=black, linestyle=solid, linewidth=0.3mm](2.21538, 33.00000)(2.33846, 33.00000)
\psline[linecolor=black, linestyle=solid, linewidth=0.3mm](2.21538, 33.00000)(2.33846, 33.00000)
\psdots[linecolor=black, linestyle=solid, linewidth=0.3mm, dotstyle=square, dotscale=1.2 1.2, linecolor=black](2.27692, 33.00000)
\rput[l](2.43077, 29.50000){\footnotesize{$E_\text{sum}(1)$}}
\psline[linecolor=darkgreen, linestyle=solid, linewidth=0.3mm](2.21538, 29.50000)(2.33846, 29.50000)
\psline[linecolor=darkgreen, linestyle=solid, linewidth=0.3mm](2.21538, 29.50000)(2.33846, 29.50000)
\psdots[linecolor=darkgreen, linestyle=solid, linewidth=0.3mm, dotstyle=triangle, dotscale=1.2 1.2, linecolor=darkgreen](2.27692, 29.50000)
}\end{pspicture}
\endgroup
    \caption{Sum energy consumption with pathloss variation}
    \label{fig:results:offloading:2}
\end{figure}

In Fig.~\ref{fig:results:offloading:2}, the total energy consumption is shown, again for both cases of $\unit[100]{\%}C_s$ and $\unit[10]{\%}C_s$, as well as for state-full and state-less offloading. The amount of total processing data $D_i$ is constant for all user devices.
The energy consumption due to in-device data processing is independent of the channel condition, hence $E_\text{sum}(0)$ is constant over $\beta$. 
In the case of full offloading, the energy consumption $E_\text{sum}(1)$ increases exponentially with increasing $\beta$. 

Consider the scenario of $\unit[100]{\%}C_s$. For low $\beta\leq2.6$, the energy consumption of $E_\text{sum}(\mathcal{A}')$ and $E_\text{sum}(1)$ are identical, because, offloading all data processing is optimal for all the user devices.
However, as $\beta$ increases, $E_\text{sum}(1)$ increases exponentially, while $E_\text{sum}(\mathcal{A}')$ does not, as only a fraction of user devices offloads. However, at all $\beta$, $E_\text{sum}(\mathcal{A}')<E_\text{sum}(0)$, which implies that strategically offloading the data processing from the user devices, can save energy. For very large $\beta$, $E_\text{sum}(\mathcal{A}')\to E_\text{sum}(0)$, because offloading data processing would become too expensive in terms of energy consumption. 

In the second scenario $\unit[10]{\%}C_s$, a similar behavior is observed, apart from $E_\text{sum}(\mathcal{A}')>E_\text{sum}(1)$ for low values of $\beta$. The reason for this behavior was already shown in Fig.~\ref{fig:results:offloading:1}, i.\,e., only a fraction of user devices can offload data processing due to limited server processing capabilities.
As the path-loss further increases, the $C_s$ is not the limiting constraint, and hence $E_\text{sum}(\mathcal{A}')$ for both the scenario converges.

Finally, as shown in Fig.~\ref{fig:results:offloading:2}, no visible differences between SF and SL offloading are observed. This implies that, at lower computational complexity, it is beneficial to either offload all data or nothing. This trend slightly changes as the amount of data elements is increased, which is discussed in the next part.

\subsection{Performance depending on data volume}
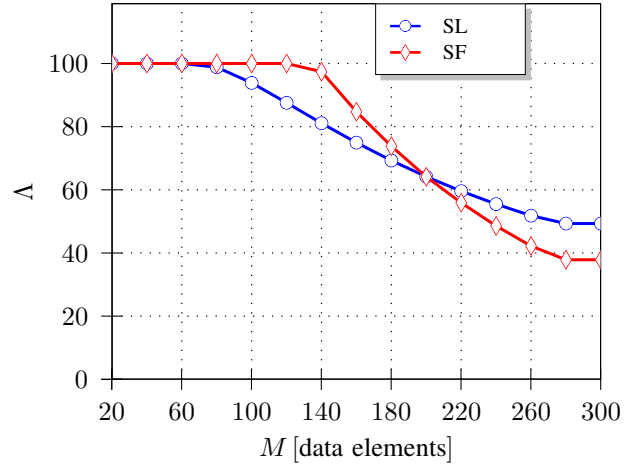
\begin{figure}[tb] %[h!]
	\centering
    \begingroup
\unitlength=1mm
\psset{xunit=0.23214mm, yunit=0.42017mm, linewidth=0.1mm}
\psset{arrowsize=2pt 3, arrowlength=1.4, arrowinset=.4}\psset{axesstyle=frame}
\begin{pspicture}(-44.61538, -38.08000)(300.00000, 119.00000)
\rput(-8.61538, -11.90000){%
\psaxes[subticks=0, labels=all, xsubticks=1, ysubticks=1, Ox=20, Oy=0, Dx=40, Dy=20]{-}(20.00000, 0.00000)(20.00000, 0.00000)(300.00000, 119.00000)%
\multips(60.00000, 0.00000)(40.00000, 0.0){6}{\psline[linecolor=black, linestyle=dotted, linewidth=0.2mm](0, 0)(0, 119.00000)}
\multips(20.00000, 20.00000)(0, 20.00000){5}{\psline[linecolor=black, linestyle=dotted, linewidth=0.2mm](0, 0)(280.00000, 0)}
\rput[b](160.00000, -26.18000){$ M \left[ \text{data elements}\right]$}
\rput[t]{90}(-36.00000, 59.50000){$\Lambda $}
\psclip{\psframe(20.00000, 0.00000)(300.00000, 119.00000)}
\psline[linecolor=blue, plotstyle=curve, linewidth=0.4mm, showpoints=true, linestyle=solid, linecolor=blue, dotstyle=o, dotscale=1.2 1.2, linewidth=0.4mm](20.00000, 100.00000)(40.00000, 100.00000)(60.00000, 100.00000)(80.00000, 98.82474)(100.00000, 93.85939)(120.00000, 87.55877)(140.00000, 81.09433)(160.00000, 74.93644)(180.00000, 69.27707)(200.00000, 64.16280)(220.00000, 59.57856)(240.00000, 55.48111)(260.00000, 51.81533)(280.00000, 49.32230)(300.00000, 49.32230)
\psline[linecolor=red, plotstyle=curve, linewidth=0.4mm, showpoints=true, linestyle=solid, linecolor=red, dotstyle=diamond, dotscale=1.2 1.2, linewidth=0.4mm](20.00000, 100.00000)(40.00000, 100.00000)(60.00000, 100.00000)(80.00000, 100.00000)(100.00000, 100.00000)(120.00000, 100.00000)(140.00000, 97.41000)(160.00000, 84.70800)(180.00000, 73.81600)(200.00000, 64.12600)(220.00000, 55.89600)(240.00000, 48.56800)(260.00000, 42.20800)(280.00000, 37.85200)(300.00000, 37.85200)
\endpsclip
\psframe[linecolor=black, fillstyle=solid, fillcolor=white, shadowcolor=lightgray, shadowsize=1mm, shadow=true](170.76923, 96.39000)(256.92308, 119.00000)
\rput[l](209.53846, 111.86000){\footnotesize{$\text{SL}$}}
\psline[linecolor=blue, linestyle=solid, linewidth=0.3mm](179.38462, 111.86000)(196.61538, 111.86000)
\psline[linecolor=blue, linestyle=solid, linewidth=0.3mm](179.38462, 111.86000)(196.61538, 111.86000)
\psdots[linecolor=blue, linestyle=solid, linewidth=0.3mm, dotstyle=o, dotscale=1.2 1.2, linecolor=blue](188.00000, 111.86000)
\rput[l](209.53846, 103.53000){\footnotesize{$\text{SF}$}}
\psline[linecolor=red, linestyle=solid, linewidth=0.3mm](179.38462, 103.53000)(196.61538, 103.53000)
\psline[linecolor=red, linestyle=solid, linewidth=0.3mm](179.38462, 103.53000)(196.61538, 103.53000)
\psdots[linecolor=red, linestyle=solid, linewidth=0.3mm, dotstyle=diamond, dotscale=1.2 1.2, linecolor=red](188.00000, 103.53000)
}\end{pspicture}
\endgroup
    \caption{Offloading with increasing computation}
    \label{fig:results:offloading:3}
\end{figure}
Fig.~\ref{fig:results:offloading:3} shows the offloading percentage, as a function of the number of data elements $M$, for both SF and SL offloading respectively.
The offloading percentage decreases with increasing $M$, as the time period $T$ stays constant, and therefore, the required spectral efficiency increases. This results in an increase of required transmit power. Hence, some user devices do not offload, as the transmit energy consumption $E_\text{tr,i}$ exceeds the in-device energy consumption $E_\text{u,i}$.

Furthermore, we can observe a clear difference between SL and SF, only at higher values of $M$. At lower $M$, $D_i$ and $C_\text{serv,i}$ is small, i.\,e. less communication and computational load is introduced. Therefore, user devices can completely offload the data in case of good channel conditions, or do not offload at all, if the channel attenuation is high. This causes the SL offloading to perform similar to the SF offloading.
In the case of SL offloading, it is possible to offload only parts of the data per user as $M$ increases. However, for SF offloading, all the data per user is either offloaded, or nothing is offloaded. Therefore, SF offloading experiences a steeper slope at higher $M$.
\begin{figure}[tb] %[h!]
	\centering
    \begingroup
\unitlength=1mm
\psset{xunit=0.23214mm, yunit=0.31250mm, linewidth=0.1mm}
\psset{arrowsize=2pt 3, arrowlength=1.4, arrowinset=.4}\psset{axesstyle=frame}
\begin{pspicture}(-44.61538, -51.20000)(300.00000, 160.00000)
\rput(-8.61538, -16.00000){%
\psaxes[subticks=0, labels=all, xsubticks=1, ysubticks=1, Ox=20, Oy=0, Dx=40, Dy=20]{-}(20.00000, 0.00000)(20.00000, 0.00000)(300.00000, 160.00000)%
\multips(60.00000, 0.00000)(40.00000, 0.0){6}{\psline[linecolor=black, linestyle=dotted, linewidth=0.2mm](0, 0)(0, 160.00000)}
\multips(20.00000, 20.00000)(0, 20.00000){7}{\psline[linecolor=black, linestyle=dotted, linewidth=0.2mm](0, 0)(280.00000, 0)}
\rput[b](160.00000, -35.20000){$M \left[ \text{data elements}\right]$}
\rput[t]{90}(-36.00000, 80.00000){$ E_\text{sum} \left[\text{mJ}\right]$}
\psclip{\psframe(20.00000, 0.00000)(300.00000, 160.00000)}
\psline[linecolor=blue, plotstyle=curve, linewidth=0.4mm, showpoints=true, linestyle=solid, linecolor=blue, dotstyle=o, dotscale=1.2 1.2, linewidth=0.4mm](20.00000, 1.12473)(40.00000, 2.60882)(60.00000, 4.56708)(80.00000, 7.13729)(100.00000, 10.29657)(120.00000, 13.90650)(140.00000, 17.85839)(160.00000, 22.06744)(180.00000, 26.46965)(200.00000, 31.01730)(220.00000, 35.67496)(240.00000, 40.41619)(260.00000, 45.22043)(280.00000, 49.48871)(300.00000, 52.02260)
\psline[linecolor=red, plotstyle=curve, linewidth=0.4mm, showpoints=true, linestyle=solid, linecolor=red, dotstyle=diamond, dotscale=1.2 1.2, linewidth=0.4mm](20.00000, 1.12473)(40.00000, 2.60882)(60.00000, 4.56708)(80.00000, 7.15103)(100.00000, 10.56057)(120.00000, 15.05948)(140.00000, 20.83186)(160.00000, 26.82347)(180.00000, 33.17149)(200.00000, 39.73277)(220.00000, 46.50419)(240.00000, 53.34676)(260.00000, 60.23785)(280.00000, 66.11309)(300.00000, 69.22049)
\psline[linecolor=black, plotstyle=curve, linewidth=0.4mm, showpoints=true, linestyle=solid, linecolor=black, dotstyle=square, dotscale=1.2 1.2, linewidth=0.4mm](20.00000, 5.00000)(40.00000, 10.00000)(60.00000, 15.00000)(80.00000, 20.00000)(100.00000, 25.00000)(120.00000, 30.00000)(140.00000, 35.00000)(160.00000, 40.00000)(180.00000, 45.00000)(200.00000, 50.00000)(220.00000, 55.00000)(240.00000, 60.00000)(260.00000, 65.00000)(280.00000, 70.00000)(300.00000, 75.00000)
\psline[linecolor=darkgreen, plotstyle=curve, linewidth=0.4mm, showpoints=true, linestyle=solid, linecolor=darkgreen, dotstyle=triangle, dotscale=1.2 1.2, linewidth=0.4mm](20.00000, 1.12473)(40.00000, 2.60882)(60.00000, 4.56708)(80.00000, 7.15103)(100.00000, 10.56057)(120.00000, 15.05948)(140.00000, 20.99584)(160.00000, 28.82890)(180.00000, 39.16469)(200.00000, 52.80285)(220.00000, 70.79850)(240.00000, 94.54392)(260.00000, 125.87617)(280.00000, 155.78542)(300.00000, 155.78542)
\endpsclip
\psframe[linecolor=black, fillstyle=solid, fillcolor=white, shadowcolor=lightgray, shadowsize=1mm, shadow=true](41.53846, 107.20000)(180.92308, 160.00000)
\rput[l](80.30769, 150.40000){\footnotesize{$E_\text{sum}(\mathcal{A}') \: \text{SL}$}}
\psline[linecolor=blue, linestyle=solid, linewidth=0.3mm](50.15385, 150.40000)(67.38462, 150.40000)
\psline[linecolor=blue, linestyle=solid, linewidth=0.3mm](50.15385, 150.40000)(67.38462, 150.40000)
\psdots[linecolor=blue, linestyle=solid, linewidth=0.3mm, dotstyle=o, dotscale=1.2 1.2, linecolor=blue](58.76923, 150.40000)
\rput[l](80.30769, 139.20000){\footnotesize{$E_\text{sum}(\mathcal{A}') \: \text{SF}$}}
\psline[linecolor=red, linestyle=solid, linewidth=0.3mm](50.15385, 139.20000)(67.38462, 139.20000)
\psline[linecolor=red, linestyle=solid, linewidth=0.3mm](50.15385, 139.20000)(67.38462, 139.20000)
\psdots[linecolor=red, linestyle=solid, linewidth=0.3mm, dotstyle=diamond, dotscale=1.2 1.2, linecolor=red](58.76923, 139.20000)
\rput[l](80.30769, 128.00000){\footnotesize{$ E_\text{sum}(0)$}}
\psline[linecolor=black, linestyle=solid, linewidth=0.3mm](50.15385, 128.00000)(67.38462, 128.00000)
\psline[linecolor=black, linestyle=solid, linewidth=0.3mm](50.15385, 128.00000)(67.38462, 128.00000)
\psdots[linecolor=black, linestyle=solid, linewidth=0.3mm, dotstyle=square, dotscale=1.2 1.2, linecolor=black](58.76923, 128.00000)
\rput[l](80.30769, 116.80000){\footnotesize{$ E_\text{sum}(1)$}}
\psline[linecolor=darkgreen, linestyle=solid, linewidth=0.3mm](50.15385, 116.80000)(67.38462, 116.80000)
\psline[linecolor=darkgreen, linestyle=solid, linewidth=0.3mm](50.15385, 116.80000)(67.38462, 116.80000)
\psdots[linecolor=darkgreen, linestyle=solid, linewidth=0.3mm, dotstyle=triangle, dotscale=1.2 1.2, linecolor=darkgreen](58.76923, 116.80000)
}\end{pspicture}
\endgroup
    \caption{Sum energy with increasing computation}
    \label{fig:results:offloading:4}
\end{figure}
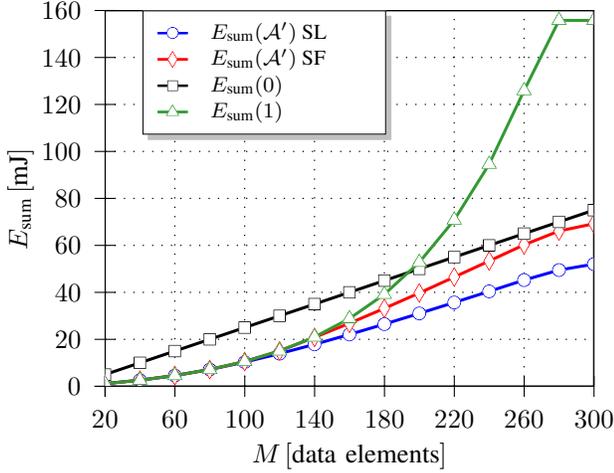
This is also reflected by the energy consumption, as shown in Fig.~ \ref{fig:results:offloading:4}. The slope with which the energy consumption increases for SF offloading, is slightly higher than in the case of SL offloading. However, both SF and SL offloading, have a lower energy consumption compared to a fully centralized processing, $E_\text{sum}(1)$, and a fully localized processing, $E_\text{sum}(0)$.
$E_\text{sum}(0)$ increases linearly because the computational complexity in our scenario scales linearly with the number of data elements $M$. In contrast, the energy consumption for fully centralized processing increases exponentially.
%Because, the transmission of the large number of data elements $M$, within the given time, will require a high spectral efficiency, and hence, more transmit power. 
% The energy consumed due to local processing is highly dependent on the complexity function of an algorithm, the size of processing data, and on the processor characteristics. Moreover, the communication channel plays an important role in making an offloading decision.
\section{Conclusion} \label{sec: conclusion}
In this paper, we developed an energy consumption model, for an in-device computation and offloading the computation. A closed form solution is obtained to optimally offload the computation, for the given cloud computational resources and channel condition. The results show that the energy consumption of the user devices can be reduced by making an informed decision, and analyzing the trade-off between the communication and computational load of the system. Furthermore, the results illustrate that the bandwidth, and the cloud server capacity, are the limiting factors to optimally offload the computation. If the processing capacity of the cloud server is limited, even with very good channel conditions, the user cannot offload to the cloud, hence, sub-optimally saving the energy. Similarly, if the system has to process a large amount of data, in a short time span, then the available bandwidth is the limiting factor. %Also, the proposed solution, further reduces the energy consumption at higher computation load. %Similarly, if the system has to support very large number of user devices, then the available bandwidth is the limiting factor. More energy can be saved at higher computation complexity by strategically offloading multiple users.\\
This paper only deals with the data processing algorithms that have linear complexity. The multi-user analytical framework can be further use to study algorithms with different complexities.
\vfill 
%Also, this paper does not consider queuing delays at the edge cloud. In future, the developed framework can be further extended by considering the queuing models.
\appendix
\subsection{Proof of Theorem 1} \label{subsec:proof1}
\begin{proof} 
In order to solve the optimization problem, we need to first apply the derivative to $\mathcal{L}$ in (\ref{eq:opt.sum.energy.10}) w.r.t. $\alpha_i$:
       \begin{equation}
    	\frac{\partial \mathcal{L}}{\partial \alpha_i} = -E_\text{u,i} + \frac{\partial E_\text{tr,i}(\alpha_i)}{\partial \alpha_i}
        	+ \nu L \cdot C_\text{serv,i}
            - \Psi_i \label{eq:opt.energy.sum.100}
    \end{equation}
If we now define the spectral efficiency $r_i = D_i/(B_i T)$ and the constant $K_i = \text{ln}(2) \left [\frac{d_i}{d_0}\right]^{\beta} N_0 D_i/G$, then
    \begin{eqnarray}
    	\frac{\partial E_\text{tr,i}(\alpha_i)}{\partial \alpha_i} & = & K_i2^{\alpha_i r_i}.
    \end{eqnarray}
    Hence, (\ref{eq:opt.energy.sum.100}) becomes
    \begin{equation} \label{eq:opt.energy.sum.120}
    	\frac{\partial \mathcal{L}}{\partial \alpha_i} = -E_\text{u,i} + K_i2^{\alpha_i r_i}
        	+ \nu L \cdot C_\text{serv,i}
            - \Psi_i,
    \end{equation}
    to which we need to apply the KKT conditions, i.\,e.,
%\begin{align} 
\begin{eqnarray}%\label{eq:opt.energy.sum.110}
    	\forall \alpha_i:  \frac{\partial \mathcal{L}}{\partial \alpha_i} = 0;
        \alpha_i \geq 0  \label{eq:opt.energy.sum.110} \\
        \sum\limits_i^N L \cdot \alpha_i \cdot C_\text{serv,i}   \leq C_\text{s, max} \label{eq:opt.energy.sum.112} \\
        \nu \: \left( \sum\limits_i^N L \cdot \alpha_i \cdot C_\text{serv,i}  - C_\text{s, max} \right) = 0 &  \label{eq:opt.energy.sum.111}\\     
        \nu  \geq 0; \Psi_i \geq 0;  \Psi_i\alpha_i  = 0 
        \end{eqnarray}
%\end{align}

In the following, we will consider four cases under which the above KKT conditions need to be considered.
\paragraph{Fully Loaded system with offloading, $\nu$ > 0, $\psi$ = 0}

In this case, we need to consider (\ref{eq:opt.energy.sum.120}) and (\ref{eq:opt.energy.sum.111}). If $\nu>0$, then
(\ref{eq:opt.energy.sum.111}) implies a fully loaded system where all resources at the edge cloud server are in use.
Now, let's focus first on (\ref{eq:opt.energy.sum.120}).
\begin{eqnarray}
 -E_\text{u,i} + K_i2^{\alpha_i r_i} +  \nu \: C_\text{serv,i} - \underbrace{\Psi_i}_{=0}  =  0 \\
  K_i2^{\alpha_i r_i}  =  E_\text{u,i} - \nu \: C_\text{serv,i} \\
 \alpha_i  =  \frac{1}{r_i}\log_2\left(\frac{1}{K_i}\left[E_\text{u,i} - \nu \: C_\text{serv,i} \right]\right) \label{eq:opt.energy.sum.200}
\end{eqnarray}
In addition, from (\ref{eq:opt.energy.sum.110}), we know that $\alpha_i \geq 0$, i\,e.,
\begin{equation} \label{eq:energy.opt.sum.210}
  \alpha_i = \left(\frac{1}{r_i}\log_2\left(\frac{1}{K_i}\left[E_\text{u,i} - \nu \: C_\text{serv,i} \right]\right)\right)^+ 
\end{equation}

\paragraph{Case 2: Underloaded System with offloading, $\nu$ = 0, $\psi$ = 0}
If $\nu$ = 0, from \eqref{eq:opt.energy.sum.111},  
\begin{equation} 
\sum\limits_i^N L \: \alpha_i \: C_\text{serv,i} - C_\text{s,max} < 0 ,
\end{equation} 
and $\psi = 0$, i.e., $\alpha_i \geq 0$. By putting $\psi$ = 0 and $\nu$ = 0 in $\eqref{eq:opt.energy.sum.120}$ we get,
\begin{eqnarray}
-E_\text{u,i} + K_i 2^{\alpha_i r_i} = 0 \\
\alpha_i = \left( \frac{1}{r_i} \cdot \log_2\left[ \frac{E_\text{u,i}}{K_i}\right] \right)^+ 
\end{eqnarray}

\paragraph{Case 3: No Offloading, $\nu$ = 0 and $\psi$ > 0}
If $\psi$ > 0, then $\alpha_i$ = 0. Using 
$\nu = 0$ and $\alpha_i = 0$ in \eqref{eq:opt.energy.sum.120} we get
\begin{eqnarray}
\psi =  K_i - E_\text{u,i}
\end{eqnarray}
And applying $\alpha_i=0$ to (\ref{eq:opt.energy.sum.112}) implies
\begin{eqnarray}
\sum\limits_i^N \alpha_i \: L\:  C_\text{serv,i} - C_\text{s,max} < 0 \\
0  < C_\text{s,max}.
\end{eqnarray}
The condition only holds as long as $\psi > 0$, i.\,e., $K_i > E_\text{u,i}$.

\paragraph{Case 4: No Offloading condition, $\nu$ > 0 and $\psi$ > 0}

If $\nu$ > 0 $\Rightarrow$ $\sum\limits_i^N \alpha_i \: C_\text{serv,i} - C_\text{s,max} = 0$. If $\psi$ > 0 $\Rightarrow$ $\alpha_i$ = 0. If $\alpha_i = 0$ in the above constraint, then this implies $0 = C_\text{s,max}$, which cannot be true.
\end{proof}

\subsection{Proof of Theorem 2} \label{subsec:proof2}
\begin{proof}
From ~\eqref{eq:opt.energy.sum.200}, $\alpha_i \geq 0$,
\begin{eqnarray} \label{eq:energy.opt.nu.ubound}
  \frac{1}{K_i}\left[E_\text{u,i} - \nu \: C_\text{serv,i} \right]  \geq  1 \\
  \frac{E_\text{u,i} - K_i}{C_\text{serv,i}}   \geq  \nu 
\end{eqnarray}
and $\alpha_i \leq 1$, according to the optimization problem in eq.~\eqref{eq:opt.sum.energy.1} and ~\eqref{eq:energy.opt.sum.210}
\begin{eqnarray} \label{eq:energy.opt.nu.lbound}
  \frac{1}{r_i}\log_2\left(\frac{1}{K_i}\left[E_\text{u,i} - \nu \: C_\text{serv,i} \right]\right)  \leq  1 \\
  \nu  \geq  \frac{E_\text{u,i} - K_i2^{r_i}}{C_\text{serv,i}}
\end{eqnarray}
From both the ~\eqref{eq:energy.opt.nu.ubound}, ~\eqref{eq:energy.opt.nu.lbound} and the fact that $\nu \geq 0$
\begin{equation}
\frac{E_\text{u,i} - K_i}{C_\text{serv,i}} \geq \nu \geq  \left[\frac{E_\text{u,i} - K_i2^{r_i}}{C_\text{serv,i}}\right]^+
\end{equation}
In a multi-user system, the bounds become
\begin{equation}
\min\limits_{\forall i: \alpha_i>0}\left[\frac{E_\text{u,i} - K_i}{C_\text{serv,i}}\right] \geq \nu \geq \max\limits_{\forall i: \alpha_i>0} \left( \left[\frac{E_\text{u,i} - K_i2^{r_i}}{C_\text{serv,i}}\right]\right)^+
\end{equation}
Note that the bound only considers user devices, which offload data processing.
\end{proof}

\balance
\bibliographystyle{IEEEtran}%{abbrv}%IEEEtran
%\bibliography{IEEEabrv,Mendeley}
%\bibliography{IEEEabrv,Globecom_2017}
\bibliography{Globecom_2017}
\end{document}